\documentclass[aps,nofootinbib,preprintnumbers,citeautoscript]{revtex4}
\usepackage{amsmath,amssymb}
\usepackage{enumerate}
\usepackage{graphicx}
\usepackage{amsthm}
 \usepackage{mathtools}
 \usepackage{textcomp}
  \usepackage{braket}
  \usepackage{url}
  \usepackage[flushleft]{threeparttable}
 \usepackage{makecell}
\usepackage{pifont}% http://ctan.org/pkg/pifont
\newcommand{\cmark}{\ding{51}}%
\newcommand{\xmark}{\ding{55}}%

\begin{document}
\title{A Quantum Finite Automata Approach to Modeling the Chemical Reactions} 

\author{Amandeep Singh Bhatia$^{1,^*}$,  Shenggen Zheng$^{2}$\\
$^{1}$\textit{Chitkara University Institute of Engineering \& Technology, Chitkara University, Punjab, India}\\
$^{2}$\textit{Center for Quantum Computing, Peng Cheng Laboratory, Shenzhen, China} \\
E-mail: $^{1,^*}$amandeepbhatia.singh@gmail.com}

\begin{abstract}
In recent years, the modeling interest has increased significantly from the molecular level to the atomic and quantum scale. The field of computational chemistry plays a significant role in designing computational models for the operation and simulation of systems ranging from atoms and molecules to industrial-scale processes. It is influenced by a tremendous increase in computing power and the efficiency of algorithms. The representation of chemical reactions  using classical automata theory in thermodynamic terms had a great influence on computer science. The study of chemical information processing   with quantum computational models is a natural goal. In this paper, we have modeled chemical reactions using two-way quantum finite automata, which are halted in linear time. Additionally, classical pushdown automata can be designed for such chemical reactions with multiple stacks. It has been proven that computational versatility can  be increased by combining chemical accept/reject signatures and quantum automata models.\\

\end{abstract}
\maketitle

\maketitle

%\newpage

%\tableofcontents

\newcommand{\myqed}{\rule{2pt}{1em}}
\newenvironment{myproof}{\begin{proof}}{\let\qedsymbol\myqed\end{proof}}
\theoremstyle{plain}
\newtheorem{thm}{Theorem}
\theoremstyle{definition}
\newtheorem{defn}{Definition}
\newtheorem{exmp}{Example}[section]
\section{Introduction \& Motivation}
Recently, the  connection between complex reactions and their thermodynamics has got overwhelming response among research communities. Initially, in the 1970s, Conrad \citep{1} processed the information of molecular systems and stated that complex biochemical systems cannot be analyzed on classical computers. Till now,  artificial approaches use complex biomolecules or logic gates based reaction-diffusion systems to solve the problems \cite{23, 24, 25}. Classical systems are not robust and capable to describe quantum systems. Some tasks that are impossible in classical systems can be realized in quantum systems. Quantum computing is concerned with computer technology based on the principles of quantum mechanics, which describes the behavior and nature of matter and energy on the quantum level \citep{2}. Quantum computation demonstrates the computation power and other properties of the computers based on principles of quantum mechanics.

Models of finite automata are abstract computing devices, which play a crucial role to solve computational problems in theoretical computer science. Classical automata theory is closely associated with formal language theory, where automata are ranked from simplest to most powerful based on their language recognition power \cite{3}. Classical automata theory has been of significant importance due to its practical real-time applications in the development of several fields. Therefore, it is a natural goal to study quantum variants of classical automata models, which play an important role in quantum information processing. 

The quantum automata theory has been developed using the principles of quantum mechanics and classical automata. Quantum computational models make it possible to examine the resources needed for computations.  Soon after the brainstorm of Shor's factorization quantum algorithm \cite{4}, the first models of quantum finite automata (QFAs) have been introduced. Initially, Kondacs and Watrous \cite{6}, and Moore and Crutchfield \cite{5} proposed the concept of quantum automata separately. There is a diversity of quantum automata models have been studied since then and demonstrated  in various directions such as quantum finite automata \cite{5}, Latvian QFA, 1.5-way QFA, two-way QFA (2QFA) \cite{6}, quantum finite state machines of matrix product states \cite{54}, quantum sequential machines \cite{58}, quantum pushdown automata, quantum Turing machine, quantum queue automata \cite{56}, quantum multicounter machines (QMCM), quantum multihead finite automata (QMFA), quantum finite automata with classical states (2QCFA) \cite{36, 37}, quantum omega automata \cite{55}, state succinctness of two-way probabilistic  finite automata (2PFA), QFA, 2QFA and 2QCFA \cite{51, 48, 49},  interactive proof systems with quantum finite automata \cite{52, 53} and semi-quantum two-way finite automata \cite{46, 47} and many more since last two decades  \citep{50, 57, 8, 9, 10}. These models are effective in determining the boundaries of various computational features and expressive power. Quantum computers are more powerful than Turing machines and even probabilistic Turing machines. Thus, mathematical models of quantum computation can be view
as generalizations of its physical models. 

Computational biochemistry has been a rapidly evolving research area at the interface between biology, chemistry, computer science, and mathematics. It helps us to apply computational models to understand biochemical and chemical processes and their properties. A combination of chemistry and classical automata theory provides a constructive means of refining the number of objects that allow to understand the energetic cost of computation \cite{11}.  The research has been consistently grown in the field of chemical computing. There exist two ways to model complex chemical reactions: abstract devices and formal models based on multiset rewriting \cite{13}. The complex chemical reaction networks carry out chemical processes that mimic the workings of classical automata models. Recently, Duenas-Diez and Perez-Mercader \cite{11, 12} have designed chemical finite automata for regular languages and chemical automata with multiple stacks for context-free and context-sensitive languages. Further, the thermodynamic  interpretation of the acceptance/rejection of chemical automata is given. It is useful to understand the energetic cost of chemical computation. They have used the one-pot reactor (mixed container), where chemical reactions and molecular recognition takes place after several steps, without utilizing any auxiliary geometrical aid.

In classical automata theory, it is known that two-way deterministic finite automata (2DFA) can be designed for all regular languages. It has also been investigated that two-way probabilistic finite automata (2PFA) can be designed for a non-regular language $L=\{a^nb^n \mid n \geq 1\}$ in an exponential time \cite{14, 15}. The research has consistently evolved in the field of quantum computation and information processing. In quantum automata theory, it has been proved that 2QFA can be designed for \textit{L} with one-sided bounded error and halted in linear time. Moreover, it has been demonstrated that 2QFA can be also designed for non-context free language $L=\{a^nb^nc^n \mid n \geq 1\}$ \cite{6}. Hence, 2QFA is strictly more powerful than its classical counterparts based on language recognition capability. 

The field of chemistry and chemical computing plays a significant role in the evolution of computational models to mimic the behavior of systems at the atomic level. It is greatly influenced by the computing power of quantum computers.
Motivated from the above-mentioned facts, we have modeled  chemical reactions in the form of formal languages and represented them using two-way quantum finite automata. The main objective is to examine how chemical reactions perform chemical sequence identification equivalent to quantum automata models without involving bio-chemistry or any auxiliary device. The crucial advantage of this approach is that chemical reactions in the form of accept/reject signatures can be processed in linear time with one-sided bounded error.  The organization of rest of this paper is as follows: Subsection is devoted to prior work. In Sect. 2, some preliminaries are given. The definition of two-way quantum finite automata is given in Sect. 3. In Sect. 4, the chemical reactions are transcribed in formal languages and modeled using two-way quantum finite automata approach. Summary of work is given in Section 5. Finally, Sect. 6 is the conclusion.

\subsection{Prior work}
The field of chemical computation has rich and interesting history. Various researchers have represented the chemical computation using the concept of logic gates based reaction-diffusion systems and artificial intelligence approaches. In early 1970s, Conrad \cite{1} differentiated the information processing on molecules from digital computing. Nearly a decade after, Okamoto et al. \cite{16} proposed the concept of a theoretical chemical diode for cyclic enzyme systems. It has been proved that it can be used to analyze the dynamic behavior of metabolic switching events in bio-computer.  In 1991, Hjelmfelt et al. \cite{17} designed neural networks and finite state machines using chemical diodes. It has been investigated that the execution of a universal Turing machine is possible using connecting chemical diodes. 

In 1995, Tóth and Showalter \cite{18} implemented AND and OR logic gates using reaction-diffusion systems, where the signals are programmed by chemical waves. It was the first empirical realization of chemical logical gates. In 1997, Magnasco \cite{19} showed that logic gates can be constructed and executed in the chemical kinetics of homogeneous solutions. It has been proved that such constructions have computational power equivalent to Turing machine. Adamatzky and Lacy Costello \cite{20} experimentally realized the Chemical XOR gate by following the same approach of Toth and Showlter in 2002. Further, Górecki et al. \cite{21}  constructed the chemical counters for an information processing in the excitable reaction-diffusion systems. 

It is one of the most promising new areas of research. Some of the difficulties can be caused by connecting several gates together for advanced computation. Thus, recently, researchers started focusing on native chemical computation, i.e. without reaction-diffusion systems. In 1994, Adleman \cite{22} proposed the concept of DNA computing and solved the Hamiltonian path problem by changing DNA strands. In 2009, Benenson \cite{23} reviewed biological measurement tools for new generation biocomputers. Prohaska et al. \cite{24} studied protein domain using chromatin computation and introduced chromatin as a powerful machine for chemical computation and information processing. In 2012, Bryant \cite{25}
proved chromatin computer as computationally universal and used to solve an example of combinatorial problem.

The structures of DNA and RNA are represented using the concept of classical automata theory \cite{26} \cite{27}. Krasinski et al.  \cite{28} represented the restricted enzyme in DNA with pushdown automata in circular mode. Khrennikov and Yurova \cite{29} modeled the behavior of protein structures using classical automata theory and investigated the resemblance between the quantum systems and modeling behavior of proteins. Bhatia and Kumar \cite{30} modeled ribonucleic acid (RNA) secondary structures using two-way quantum finite automata, which are halted in linear time.  Recently, Duenas-Diez and Perez-Mercader designed molecular machines for chemical reactions. The native chemical computation has been implemented beyond the scope of logic gates i.e. with chemical automata \cite{12}. It has been demonstrated that chemical reactions transcribed in formal languages, can be recognized by Turing machine without using biochemistry \cite{11}.

\section{Preliminaries}
In this section, some preliminaries are given. We assume that the reader is familiar with the classical automata theory and the concept of quantum computing; otherwise, reader can refer to the theory of automata \cite{3}, quantum information and computation \cite{2} \cite{7}. Linear algebra is inherited from quantum mechanics to describe the field of quantum computation. It is a crucial mathematical tool. It allows us to represent the quantum operations and quantum states by matrices and vectors, respectively, that obey the rules of linear algebra. The following are the notions of linear algebra used in quantum computational theory:
\begin{itemize}
	\item Vector space (\textit{V}) \cite{7}: A vector space (\textit{V}) is defined over the field $\mathbb{F}$ of complex numbers $\mathbb{C}$ consisting of a non-empty set of
vectors, satisfying the following operations:
\begin{itemize}
\item Addition: If two vectors $\ket{a}$ and $\ket{b}$ belong to \textit{V}, then $\ket{a}+\ket{b} \in V$.
\item Multiplication by a scalar: If $\ket{a}$ belongs to \textit{V}, then $\lambda\ket{a} \in V$, where $\lambda \in \mathbb{C}$. 
\end{itemize}

	\item Dirac notation \cite{2}: In quantum mechanics, the Dirac notation is one of the most peculiarities of linear algebra. The combination of vertical and angle bars ($\ket{} \bra{}$) is used  to unfold quantum states. It provides an inner product of any two vectors.  The bra $\bra{b}$ and ket $\ket{a}$ represent the row vector and column vector, respectively.
	\begin{equation}
	\ket{a}=\begin{bmatrix}
	\alpha_1\\
	\alpha_2\\
	\alpha_3\\
	\end{bmatrix}, \bra{b}=\begin{bmatrix}
	\beta_1^* & \beta_2^* & \beta_3^*
	\end{bmatrix}, \ket{a}\bra{b}=\begin{bmatrix}
	\alpha_1\beta_1^* &  \alpha_1\beta_2^*  & \alpha_1\beta_3^*\\
	\alpha_2\beta_1^* &  \alpha_2\beta_2^*  & \alpha_2\beta_3^*\\
	\alpha_3\beta_1^* &  \alpha_3\beta_2^*  & \alpha_3\beta_3^*\\
	\end{bmatrix} \end{equation}
	
	where $\beta_i^*$ indicates the complex conjugate of complex number $\alpha_i$.

	\item Quantum bit \cite{10}: A quantum bit (qubit) is a unit vector defined over complex vector space $\mathbb{C}^2$. In general, it is represented as a superposition of two basis states labeled $\ket{0}$ and $\ket{1}$. 
	\begin{equation}
	\ket{\phi}=\alpha\ket{0}+\beta\ket{1} 
	\end{equation}
 The probability of state occurrence $\ket{0}$ is $|\alpha|^2$ and $\ket{1}$ is $|\beta|^2$. It satisfies that  $|\alpha|^2+ |\beta|^2=1$. The two complex amplitudes ($\alpha$ and $\beta$) are represented by one qubit. Thus, $2^n$ complex amplitudes can be represented by $\it n$ qubits. 
 
	\item Quantum state \cite{2}: A quantum state $\ket{\psi}$ is defined as a superposition of classical states, 
	\begin{equation}
	\ket{\psi}=\alpha_1\ket{w_1}+\alpha_2\ket{w_2}+...+\alpha_n\ket{w_n} \end{equation}
	where $\alpha_i's$ are complex amplitudes and $\ket{w_i}'s$ are classical states for $1 \leq i\leq n$. Therefore, a quantum state $\ket{\psi}$ can be represented as \textit{n}-dimensional column vector. 
	\begin{equation}
	\begin{bmatrix}
	\alpha_1\\
	\alpha_2\\
	...\\
	\alpha_n\\
	\end{bmatrix}
	\end{equation}

	\item Unitary transformation: In quantum mechanics, the transformation between the quantum systems must be unitary. Consider a state $\ket{\psi}$ of quantum system at time \textit{t}: $\ket{\psi}=\alpha_1\ket{w_1}+\alpha_2\ket{w_2}+...+\alpha_n\ket{w_n}$ transformed into state $\ket{\psi'}$ at time \textit{t'}: $\ket{\psi'}=\alpha_1'\ket{w_1}+\alpha_2'\ket{w_2}+...+\alpha_n'\ket{w_n}$, where complex amplitudes are associated by $\ket{\psi'(t')}=U(t'-t)\ket{\psi(t)}$, where \textit{U} denotes a time reliant unitary operator satisfies that $(U(t'-t))^*U(t'-t)=1$, and $\sum_{i=1}^{n} |\alpha_i|^2= |\alpha'_i|^2=1$ \cite{2}.

\item Hilbert space: A physical system is described by a complex vector space called Hilbert space $\mathcal{H}$ \cite{7}. It allow us to describe the basis of the quantum system. The direct sum $\braket{x|y}: \mathcal{H} \oplus \mathcal{H}\rightarrow \mathbb{C}$  or inner product $\braket{x|v}: \mathcal{H} \otimes \mathcal{H}\rightarrow \mathbb{C}$  of two subspaces, satisfying the following properties for any vectors:
	\begin{itemize}
		\item Linearity: ($\alpha\bra{x}+\beta\bra{y}$)$\ket{z}=\alpha\braket{x|z}+\beta\braket{y|z}$
		\item Symmetric property: $\braket{x|y}=\braket{y|x}$
		\item Positivity:  $\braket{x|x}\geq0$ and $\braket{x|x}=0$ iff $x=0$, where $x \in \mathcal{H}$.
	\end{itemize}
	where $x,y,z \in \mathcal{H}$ and $ \alpha, \beta \in \mathbb{C}$.

\item Quantum finite automata (QFA) \cite{34}: It is defined as a quadruple ($Q, \Sigma, s_{int}, U_{\sigma}$), where
	\begin{itemize}
		\item \textit{Q} is a set of states,
		\item $\Sigma$ is an input alphabet, 
		\item Hilbert space $\mathcal{H}$ and $s_{init} \in \mathcal{H}$ is an initial vector such that $|s_{init}|^2=1$,
		\item $\mathcal{H}_{acc} \subset \mathcal{H}$ and $P_{acc}$ is an acceptance projection operator on $\mathcal{H}_{acc}$, 
		\item $U_{\sigma}$ denotes a unitary transition matrix for each input symbol ($\sigma \in \Sigma$). 
	\end{itemize}
	
The computation procedure of QFA consists of an input string $w=\sigma_n\sigma_2...\sigma_n$. The automaton works by reading each input symbol and their respective unitary matrices are applied on the current state, starting with an initial state. The quantum language accepted by QFA is represented as a function $f_{QFA}(w)=|s_{init}U_w P_{acc}|^2$, where $U_w=U_{\sigma_1}U_{\sigma_2}...U_{\sigma_n}$. The tape head is allowed to move only in the right direction. Finally, the probability of QFA being in an acceptance state is observed; i.e. indicating whether the input string is accepted or rejected by QFA. It is also called as a real-time quantum finite automaton. 
\end{itemize}
Based on the movement of tape head, QFA is classified as one-way QFA, 1.5-way QFA, and 2QFA. In 1.5-way QFA, the tape head is permitted to move only in the right direction or can be stationary, but it cannot move towards the left direction. It has been proved that it can be designed for non-context free languages, if the input tape is circular \cite{35}. In this paper, we focused on the 2QFA model due to its more computational power than classical counterparts. 

\section{Two-way quantum Finite automata}

A quantum finite automaton (QFA) is a quantum variant of a classical finite automaton. In QFA, quantum transitions are applied on reading the input symbols from the tape \cite{5}. Two-way quantum finite automaton (2QFA) is a quantum counterpart of a two-way deterministic finite automaton (2DFA). In 2QFA, the tape head is allowed to move either in the left or right direction or can be stationary. The illustration of 2DFA is shown in Fig 1.  It consists of a read-only
input tape consisting of a sequence of cells, each one 
storing a input symbol. The read-only input tape is scanned by an input head, which is allowed to move in both  directions or can be stationary. At each step of a
computation, a finite state control should be in state from a finite set \textit{Q}.

\begin{figure}[h]
	\centering
	\includegraphics[scale=1.2]{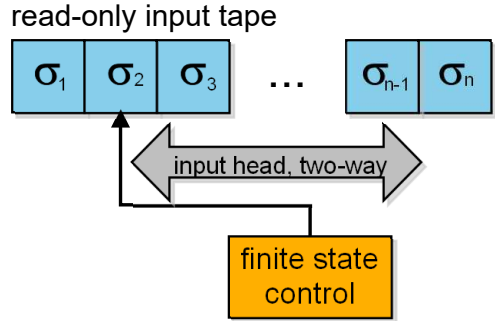}
	\caption{Schematic diagram of the “hardware” of 2DFA}
\end{figure}

\begin{defn} \emph{\cite{5}}
	A two-way quantum finite automaton is represented as sextuple $(Q, \Sigma, \delta,q_0, Q_{acc},$ $Q_{rej})$, where
	\begin{itemize}
		\item $\it Q$ is a finite set of states.
		\item $\Sigma$ is an input alphabet,
		\item Transition function $\delta$ is defined by $\delta: Q\times\Gamma\times Q\times D \rightarrow \mathbb{C}$, where $\mathbb{C}$ is a complex number,  $\Gamma=\Sigma \cup \{\#, \$\}$ and $D=\{-1, 0, +1\}$ represent the left, stationary and right direction of tape head.
\item $Q=Q_{acc}\cup Q_{rej}\cup Q_{non}$, where $Q_{non}, Q_{acc}, Q_{rej}$ represent the set of non-halting, accepting, rejecting states respectively. 	 The transition function must satisfies the following conditions:
	\end{itemize}
	\begin{enumerate}[(i)]
			\item Local probability and orthogonality condition:
			$$ \sum_{(q', d)\in Q\times D}^{\forall(q_1,\sigma_1),(q_2, \sigma_2)\in Q \times\Gamma} 	\overline{\delta(q_1,\sigma,q',d)}\delta(q_2,\sigma,q',d)=
			\left\{
			\begin{array}{ll}
			
			1~ ~ q_1=q_2\\
			0~ ~ q_1\neq q_2\\
			
			\end{array} \right \} $$
			\item First separability condition:
			$$ \sum_{q'\in Q}^{\forall(q_1,\sigma_1),(q_2, \sigma_2)\in Q \times\Gamma} 	\overline{\delta(q_1,\sigma_1,q',+1)}\delta(q_2,\sigma_2,q',0)+\overline{\delta(q_1,\sigma_1,q',0)}\delta(q_2,\sigma_2,q',-1)=0$$
			\item Second separability condition:
			$$ \sum_{q'\in Q}^{\forall(q_1,\sigma_1),(q_2, \sigma_2)\in Q \times\Gamma} 	\overline{\delta(q_1,\sigma_1,q',+1)}\delta(q_2,\sigma_2,q',-1)=0$$
			
		\end{enumerate}

\end{defn}
For each $\sigma \in \Gamma$, a 2QFA is said to be simplified, if there exists a unitary linear operator $V_\sigma$ on the inner product space such that $L_2\{Q\}\rightarrow L_2\{Q\}$. The transition function is represented as
	\begin{equation}
\delta(q,\sigma,q',d)= \left\{ \begin{array}{l}
q'V_{\sigma}q\\
0 \end{array}
\middle\vert\;
\begin{array}{@{}l@{}}
\text{if} ~ D(q')=d \\
\text{else}
\end{array}
\right\}
\end{equation}
where $q'V_{\sigma} q$ is a coefficient of $\ket{q'}$ in  $V_\sigma \ket{q}$.

Consider an input string $\it w$, written on the input tape enclosed with both end-markers such as $\#w\$$. The computation of 2QFA  is progress as follows. The tape head is above the input symbol $\sigma$ and the automaton is in any state $\it q$. Then, the state of 2QFA is changed to $q'$ with an amplitude $\delta(q,\sigma,q',d)$ and moves the tape head one cell towards right, stationary and in left direction according to $\in \{-1, 0, +1\}$. It corresponds to a unitary evolution in the inner-product space $\mathcal{H}_n$.

A computation of a 2QFA  is a chain of superpositions $c_0,c_1,c_2,....,$ where $c_0$ denotes an initial configuration. For any $c_i$, when the automaton is observed in a superposition state with an amplitude $\alpha_c$, then it has the form $U_\delta\ket{c_i}\sum_{c\in C_n}\alpha_c\ket{c_i}$, where $C_n$ represents the set of configurations. The probability associated with a configuration is calculated by absolute squares of amplitude. Superposition is said to be valid; 
if the sum of the squared moduli of their probability amplitudes is unitary.
In quantum theory, the time evolution is specified by unitary transformations. Each transition function $\delta$ prompts a transformation operator over the Hilbert space $\mathcal{H}_n$ in linear time.
$$ U_{\delta}^w \ket{q,j}=\sum_{(q',d)\in Q\times D}\delta(q,w(j),q',d)\ket{q',j+d ~ mod\lvert w \rvert}$$
for each $(q,j)\in C_{\lvert w\rvert}$, where $q\in Q, j\in Z_{\lvert w\rvert}$ and extended to $\mathcal{H}_n$ by linearity \cite{5, 31}. 

\section{Modeling of Chemical Reactions}

\begin{figure}[h]
	\centering
	\includegraphics[scale=0.7]{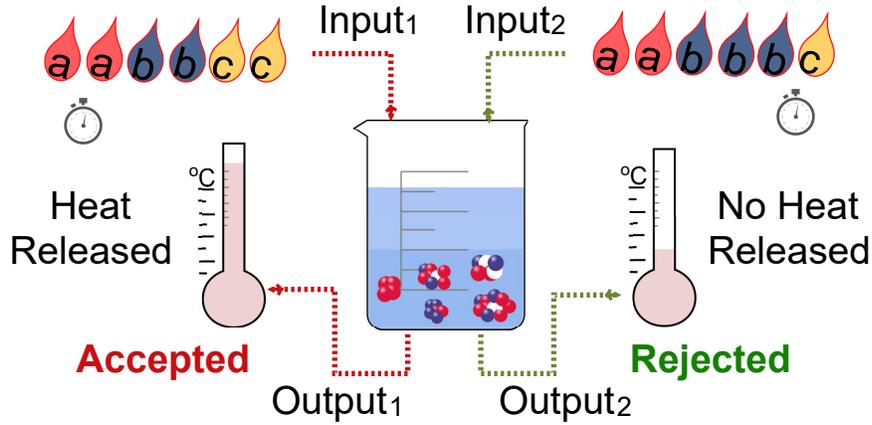}
	\caption{Representation of language recognition by chemical computation. It is reproduced from \cite{12} under the Creative Commons CCBY license.}
\end{figure}
Before, we recognize the chemical reactions using two-way quantum finite automata model, it is important to show how computational chemistry works. Fig 2. shows the illustration of language recognition by the chemical computation model. It consists of three parts: (i) a mixed container where the computation process occurs, (ii) an input translator that translates the chemical aliquots into input symbols and gives them consecutively depending upon the processing time, (3) a system to monitor the response of an automaton as a chemical criterion. Finally, the chemical computation produces well-defined chemical accept/ reject signatures for the input. For instance, if the number of \textit{a}'s and \textit{b}'s are equal in the input, then the chemical computation produces heat i.e. an input is said to be accepted. Otherwise, if no heat is released at the end of computation, then the input is said to be rejected by the system. The following are the construction of two-way quantum finite state machines of chemical reactions. 

\begin{thm}
	Two-way quantum finite automata can recognize all regular languages. 
\end{thm}
\begin{proof}
	The proof has been shown in \cite{5}.
\end{proof}

\subsection{Chemical reaction-1 consisting regular language}

For an illustrative and visual implementation, we
can choose a precipitation reaction in an aqueous medium such as
\begin{equation}
KIO_3 +AgNO_3 \rightarrow AgIO_3(s) +KNO_3
\end{equation}

If, during computation, a white precipitate of silver iodate is observed,
then the input string is said to be accepted; if the solution is clear
from precipitate, the string has been rejected because there was no reaction. Therefore, we have chosen the recipes of alphabet symbols \textit{a} for potassium iodate ($KIO_3$) and \textit{b} for silver nitrate ($AgNO_3$), quantitatively. Fig. 3 shows the chemical representation of symbols \textit{a} and \textit{b},
the bimolecular precipitation reaction \cite{11}. If the precipitate AgIO$_3$ is not presented in the solution, then the computation is said to be rejected. For example, the input string $w=aaab$ is said to be accepted due to presence of precipitate or, equally, the heat has been determined during computation. But, the input $w=aa$ is said to be rejected due to absence of precipitate, or, precisely, the heat has not been observed. The Kleene star ($\Sigma^*$) operator is a set of infinite strings of all lengths over input alphabet as well as empty string ($\epsilon$). Fig. 4 shows the corresponding theoretical 2QFA state transition graph to recognize $L_1$. 

\begin{figure}[h]
	\centering
	\includegraphics[scale=0.55]{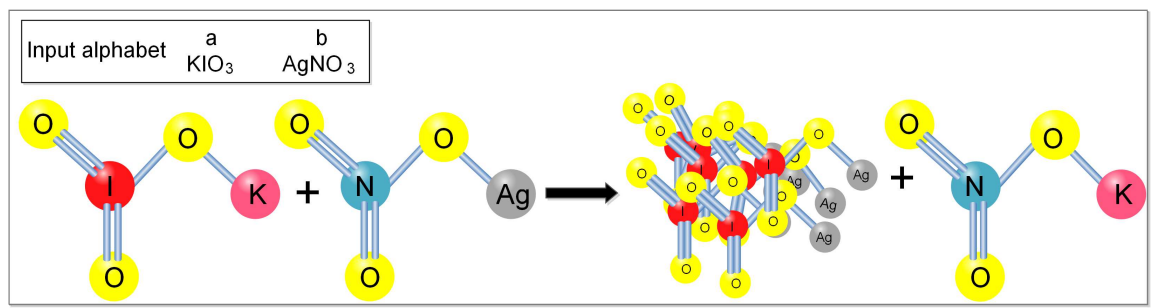}
	\caption{Illustration of acid/base reaction of $L_1$}
\end{figure}

\begin{figure}[h]
	\centering
	\includegraphics[scale=0.42]{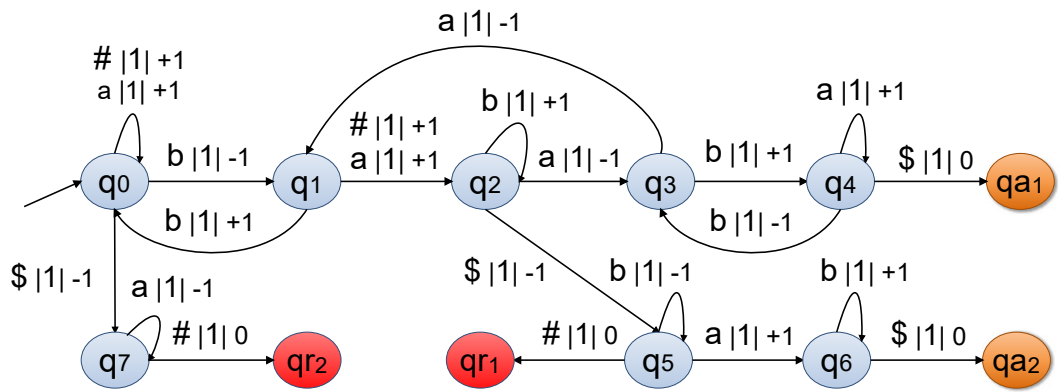}
	\caption{State transition diagram of $L_1$}
\end{figure}

\begin{thm}
A language $L_{1}=\{(a+b)^*a(a+b)^*b(a+b)^*aa^*bb^*\}$ representing precipitation reaction in Eq. (6) can be recognized by 2QFA. 
 
\end{thm}
\begin{proof}
The idea of this proof is as follows. The initial state $q_0$ reads a right-end marker $\#$ and moves the head towards the right direction. If there is no occurrence of symbol b, then it shows no precipitate,
and the input is said to be rejected by the 2QFA. Similarly, on reading the symbol \textit{b}, the state $q_0$ is changed into $q_1$. If there is no occurrence of symbol \textit{a}, then the state is transformed into rejecting state $q_{r_1}$. If the input string $w \in L_1$ contains at
least one \textit{a} and one \textit{b}, then silver iodate is present during computation and it is said to be recognized by 2QFA.  A 2QFA for $L_1$ is defined as follows: $$M_{2QFA}=(Q, \Sigma, q_0, Q_{acc}, Q_{rej}, \delta ),$$ where
\begin{itemize}
 \item[--] $Q=\{q_0, q_1, q_2, q_3, q_4, q_5, q_6, q_7, q_{a_1}, q_{a_2}, q_{r_1}, q_{r_2}\},$ where $q_0$ and $q_2$ are used to move the head towards the \$ on reading \textit{a}'s and \textit{b}'s respectively. The states $q_1$ and $q_3$ are used to confirm the last symbol read by head is \textit{a} and \textit{b}, respectively. .
	\item[--] $ \Sigma=\{a, b\}, q_0$ is an initial state, $Q_{acc}=\{q_{a_1}, q_{a_2} \}$ and $Q_{rej}=\{q_{r_1}, q_{r_2}\}$.
	\item[--] The specification of transition functions are given in Table 1. 
\end{itemize}   
It can be noted that in 2QFA where transition matrices consist 0 and 1, i.e. basically a two-way reversible finite automata (2RFA). Therefore, 2QFA can be designed for all the languages accepted by 2RFA. In transition matrix, each column and row have exactly only one entry 1. Hence, the dot product of any two rows is equal to zero. It is known that the language recognition power of 2RFA is an equivalent to 2DFA.
\begin{table}[!ht]
\centering
\caption{Details of the transition functions and head function for $L_1$}
\begin{tabular}{ |p{2.6cm}|p{2.6cm}|p{2.5cm}|p{2.5cm}|}
	
	\hline
	$V_{\#}\ket{q_0}=\ket{q_0}$ & $V_{a}\ket{q_0}=\ket{q_0}$ & $V_{b}\ket{q_0}=\ket{q_1}$ & $ V_{\#}\ket{q_1}=\ket{q_2} $ \\
	\hline
	$V_{\$}\ket{q_0}=\ket{q_7}$ & $V_{a}\ket{q_1}=\ket{q_2}$ & $V_{b}\ket{q_1}=\ket{q_2}$ & $ V_{b}\ket{q_2}=\ket{q_2} $ \\
	\hline
	 $V_{\$}\ket{q_7}=\ket{q_{r_2}}$ & $V_{a}\ket{q_7}=\ket{q_7}$ & $V_{a}\ket{q_2}=\ket{q_3}$ & $ V_{a}\ket{q_3}=\ket{q_1}$ \\
	\hline
	$V_{\$}\ket{q_2}=\ket{q_5}$ & $V_{a}\ket{q_4}=\ket{q_4}$ & $V_{b}\ket{q_3}=\ket{q_4}$ & $ V_{b}\ket{q_4}=\ket{q_3}$ \\
	\hline
	$V_{\#}\ket{q_5}=\ket{q_{r_1}}$ & $V_{a}\ket{q_5}=\ket{q_6}$ & $V_{b}\ket{q_5}=\ket{q_5}$ & $ V_{b}\ket{q_6}=\ket{q_6} $ \\
	\hline
	$V_{\$}\ket{q_4}=\ket{q_{a_1}}$ & $V_{a}\ket{q_6}=\ket{q_{a_2}}$ & & \\
	\hline
	\multicolumn{4}{|l|}{Head Functions:} \\
	\hline
	\multicolumn{4}{|c|}{$D(q_0)= (+1), D(q_1)= (-1), D(q_2)= (+1),
		D(q_3)= (-1),$} \\
	
	\multicolumn{4}{|c|}{$D(q_4)= (+1), D(q_5)= (-1), D(q_6)= (+1), D(q_7)= (-1),$} \\
	
	\multicolumn{4}{|c|}{$D(q_{a_1})=D(q_{a_2})=(0), D(q_{r_1})=D(q_{r_2})= (0)$} \\
	\hline
	
\end{tabular}
\end{table}

\end{proof}
\subsection{Chemical reaction-2 consisting context-free language}

Next, we have considered the context-free language from Chomsky hierarchy satisfying the balanced chemical reaction between NaOH and malonic acid is as follows:
\begin{equation}
H_2C_3H_2O_4 + 2NaOH \rightarrow Na_2C_3H_2O_4 + 2H_2O
\end{equation}
\begin{figure}[!ht]
	\centering
	\includegraphics[scale=0.5]{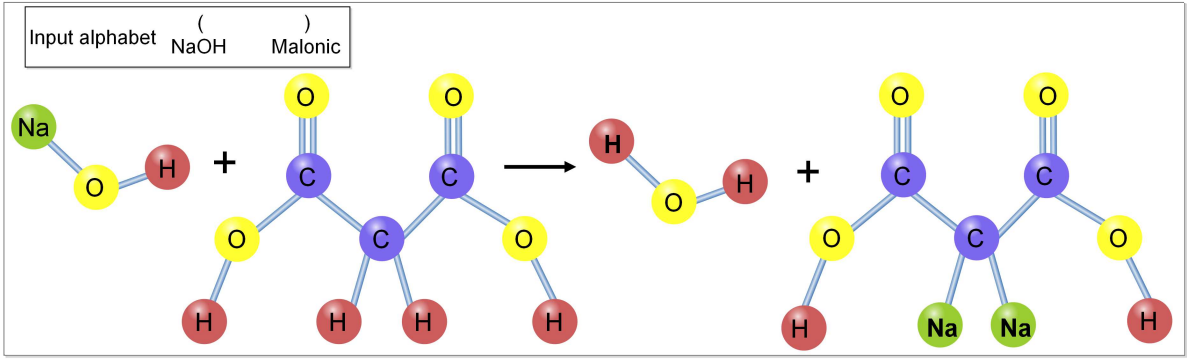}
	\caption{Illustration of the acid/base reaction of $L_2$}
\end{figure}
The language generated by the above-mentioned chemical reaction is $L_{2}$ consisting Dyck language of all words with balanced parentheses. 2QFA is designed for $L_2$ is as follows:

\begin{figure}[!ht]
	\centering
	\includegraphics[scale=0.56]{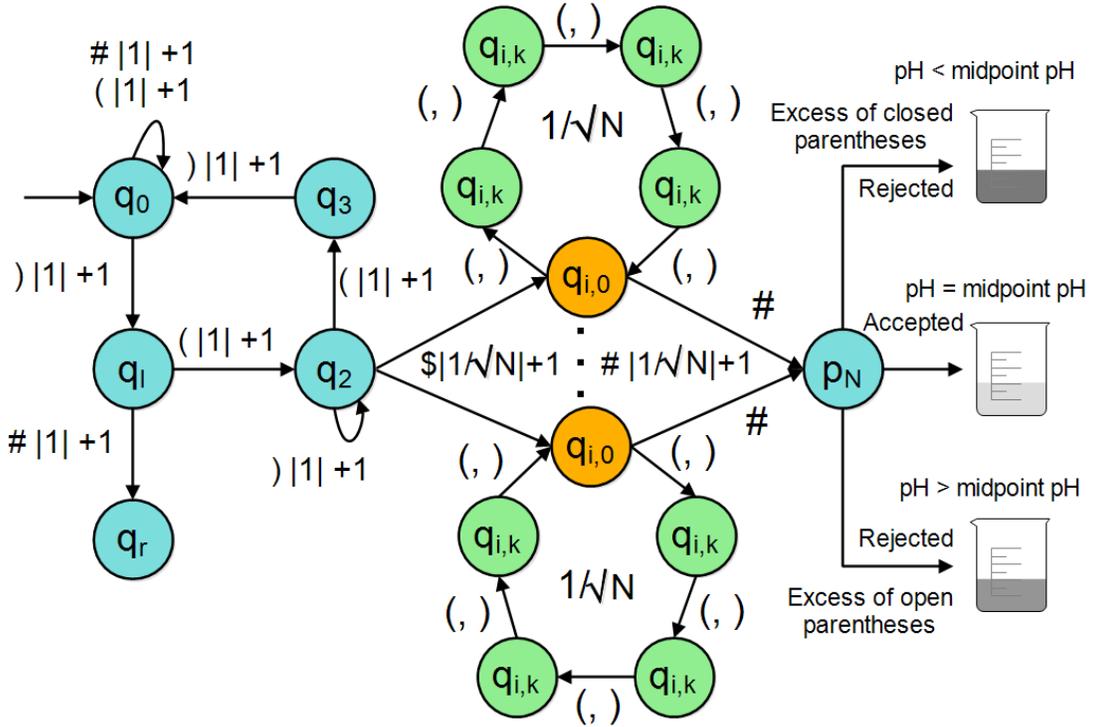}
	\caption{State transition diagram of $L_2$}
\end{figure}

\begin{thm}
A language $L_{2}$ consisting Dyck language of all words with balanced parentheses can be recognized by 2QFA with probability 1, otherwise rejected with probability at least $1-\dfrac{1}{N}$, where \textit{N} is any positive number. 
\end{thm}
\begin{proof}
The idea of this proof is as follows. It consists of three phases. First, the initial state $q_0$ reads a first symbol and both heads start moving towards the right-end marker \$. If the input string starts with  closed parentheses, then it is said to be rejected. On reading the left-end marker $\#$, the computation is split into \textit{N} paths, denoted by $q_{1, 0}, q_{2, 0}, ..., q_{N, 0}$. Each path possesses equal amplitude $\dfrac{1}{\sqrt{N}}$. Along the \textit{N }different paths, each path moves deterministically to the right-end marker $\$$.
Each computational path keep track of the open parentheses with respect to the closed parentheses. If at the end of computation, the excess of open parentheses is observed, then it is said to be rejected. It means pH value is greater than midpoint pH and intermediate gray tone is observed. Secondly, if there is an excess of closed parentheses, then the darkest gray tone is observed i.e. pH value is less than midpoint pH. It is said to be rejected by 2QFA with probability $1-\dfrac{1}{N}$. If there is a balanced occurrence of open and closed parentheses, the input string is said to be accepted with probability 1. Hence, pH value is equal to midpoint pH and the lightest gray tone is observed at the end of computation. A 2QFA for $L_2$ is defined as follows: $$M_{2QFA}=(Q, \Sigma, q_0, Q_{acc}, Q_{rej}, \delta ),$$ where
\begin{itemize}
\item[--] $Q=\{q_0, q_1, q_2, q_3\} \cup \{q_{i, j} \mid 1 \leq i \leq N, 0 \leq j \leq \text{max}(i, N-i+1)\} \cup \{p_k \mid 1 \leq k \leq N\}\cup $ $\{s_{i, 0}, w_{i, 0}, r_{i, 0} \mid 1 \leq i \leq N\} \cup \{q_{acc}, q_{rej}\}$, where $q_1$ is used to check whether the first symbol is an open parentheses or not, $q_2$ and $q_3$ are used to traverse the input string.
	\item[--] $ \Sigma=\{(, )\}, q_0$ is an initial state, $Q_{acc}=\{p_N\}$ and $Q_{rej}=\{q_r\} \cup \{p_k \mid 1 \leq k < N\} \cup$ $\{r_{i, 0} \mid 1 \leq i \leq N\}$.
	\item[--] The specification of transition functions is given in Table 2.
\end{itemize}   
\end{proof}

\begin{table} [!ht]
\centering
\caption{Details of the transition functions and head function for $L_2$}
\begin{tabular}{|p{8.25cm} p{8.25cm}|}
	
	\hline
	\multicolumn{2}{|c|}{$V_{\#}\ket{q_0}=\ket{q_0}, V_{(}\ket{q_0}=\ket{q_0}, V_{)}\ket{q_0}=\ket{q_1}, V_{\#}\ket{q_1}=\ket{q_r}
$}\\
	\hline
	\multicolumn{2}{|c|}{$V_{(}\ket{q_1}=\ket{q_2}, V_{)}\ket{q_2}=\ket{q_2}, V_{(}\ket{q_2}=\ket{q_3}, V_{)}\ket{q_3}=\ket{q_0}$}\\
	\hline
\multicolumn{2}{|c|}{$V_{\#}\ket{q_2}=\dfrac{1}{\sqrt{N}} \sum_{i=1}^{N}{\ket{q_{i, 0}}}$}	 \\
	\hline
	\multicolumn{2}{|c|}{$V_{(}\ket{q_{i, 0}}=\ket{q_{i, i}},~ ~ V_{)}\ket{q_{i, 0}}=\ket{q_{i, N-i+1}}, \text{for} ~ 1 \leq i \leq N, V_{(}\ket{q_{i, j}}=\ket{q_{i, j-1}}, \text{for} ~ 1 \leq j \leq i$} \\
	\hline
	 \multicolumn{2}{|c|}{$V_{)}\ket{q_{i, j}}=\ket{q_{i, j-1}}, \text{for} ~ 1 \leq j \leq N-i+1, 1 \leq i \leq N$} \\
	\hline
	\multicolumn{2}{|c|}{$V_{\$}\ket{q_{i, 0}}= \ket{s_{i,0}}, V_{)}\ket{s_{i, 0}}= \ket{w_{i,0}}, V_{(}\ket{s_{i, 0}}= \ket{r_{i, 0}}, \text{for} ~1 \leq i \leq N $}\\
	\hline
	\multicolumn{2}{|c|}{$V_{\$}\ket{w_{i, 0}}= \dfrac{1}{\sqrt{N}} \sum_{k=1}^{N}{exp\left(\dfrac{2\pi i}{N}ki \right) \ket{p_k}}, \text{for} ~1 \leq i \leq N$}  \\
	\hline
	\multicolumn{2}{|l|}{Head Functions:} \\
	\hline
	\multicolumn{2}{|c|}{$D(q_0)= (+1), D(q_1)= (-1), D(q_2)= (+1), D(q_3)= (-1), D(q_{i, 0})= (+1),  \text{for} ~1 \leq i \leq N $}\\
	
	\multicolumn{2}{|c|}{$D(q_{i, j})= (0), \text{for} ~1 \leq i \leq N,j \neq 0, D(q_{r})= (0),  D(p_{k})= (0), \text{for} ~1 \leq k \leq N$} \\
	
	\multicolumn{2}{|c|}{$D(s_{i, 0})= (-1), D(w_{i, 0})= (+1), D(r_{i, 0})= (0), \text{for} ~1 \leq i \leq N$} \\

	\hline
	
\end{tabular}
\end{table}
\subsection{Chemical reaction-3 consisting context-sensitive language}

To implement a chemical 2QFA for context-sensitive language, we have used  Belousov-Zhabotinsky (BZ) reaction network for the non-linear oscillatory
chemistry \cite{11}, which consists temporal oscillation in the sodium bromate and malonic acid system \cite{33} as
\begin{equation}
3BrO_3- + 5CH_2(COOH)_2 + 3H^+ \rightarrow 3BrCH(COOH)_2 + 4CO_2 + 2HCOOH + 5H_2O
\end{equation}
\begin{figure}[!ht]
	\centering
	\includegraphics[scale=0.35]{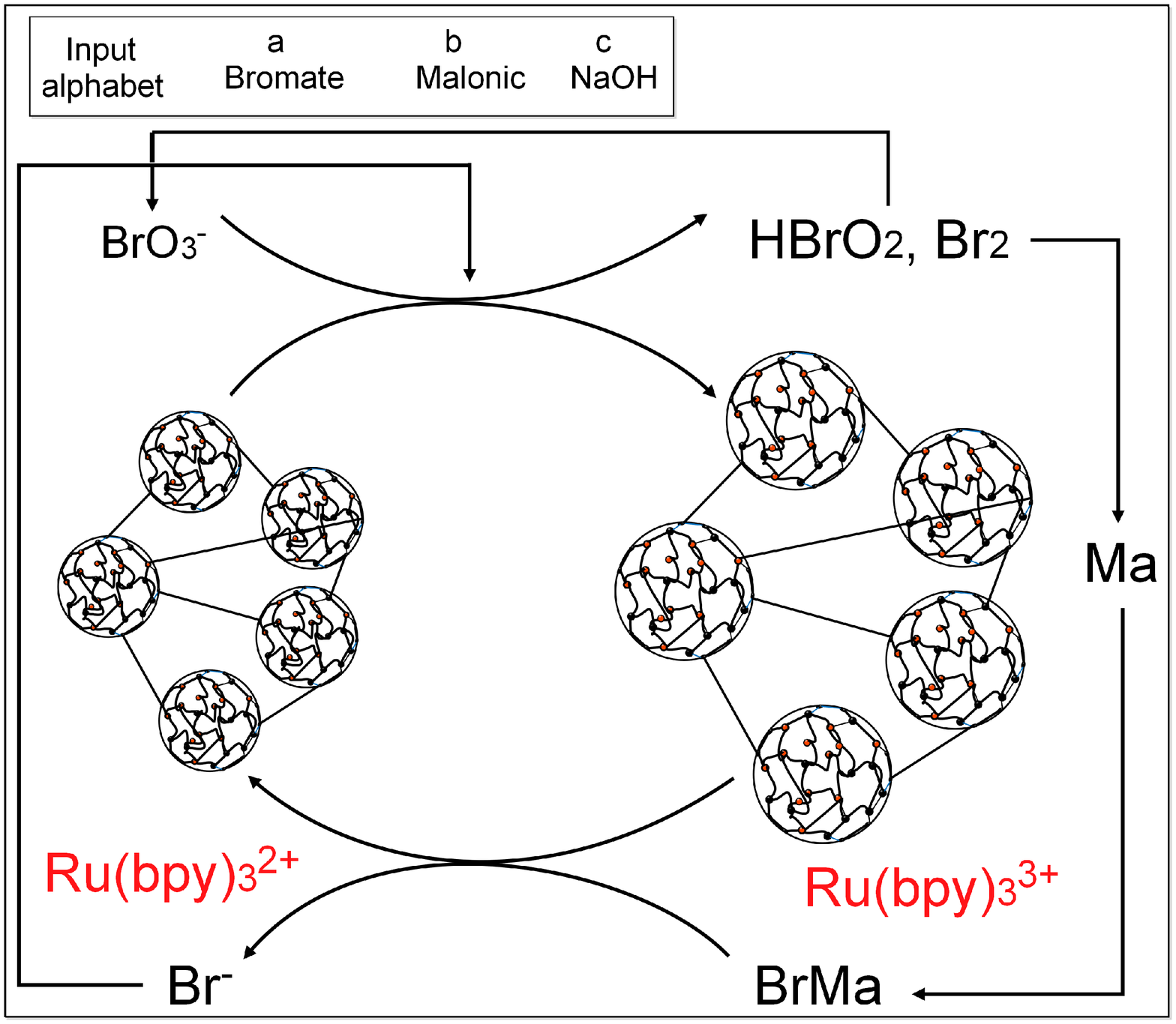}
	\caption{Illustration of the acid/base reaction of $L_3$}
\end{figure}

In 2019, Duenas-Diez and Perez-Mercader \cite{11} designed chemical Turing machine for BZ reaction network. The chemical reaction is fed sequentially to the reactor as $\{(BrO_3^-)^n (MA)^n (NaOH)^n\}$, where \textit{n}$>$0. It is transcribed in formal language as $L_{3}=\{a^nb^nc^n | n > 0\}$. The symbol \textit{a} is interpreted as a fraction of sodium bromate, \textit{b} is used for
malonic acid and symbol \textit{c} is transcribed as a quantity of NaOH.  
It is known that $L_3$ is a context-sensitive language and cannot be recognized by finite automata or pushdown automata with a stack. Although, it can be recognized by two-stack PDA. We have shown that $L_3$ can be recognized by 2QFA  without using any external aid.

\begin{thm}
A language $L_{3}=\{a^nb^nc^n | n > 0\}$ can be recognized by 2QFA in linear time. 
For a language $L_{3}=\{a^nb^nc^n | n > 0\}$ and for arbitrary \textit{N}-computational paths, there exists a 2QFA such that for $w \in L_3$, it accepts \textit{w} with bounded error $\epsilon $ and rejects $w \notin L_3$ with probability at least $1-\dfrac{1}{N}$.
\end{thm}

\begin{figure}[!ht]
	\centering
	\includegraphics[scale=0.7]{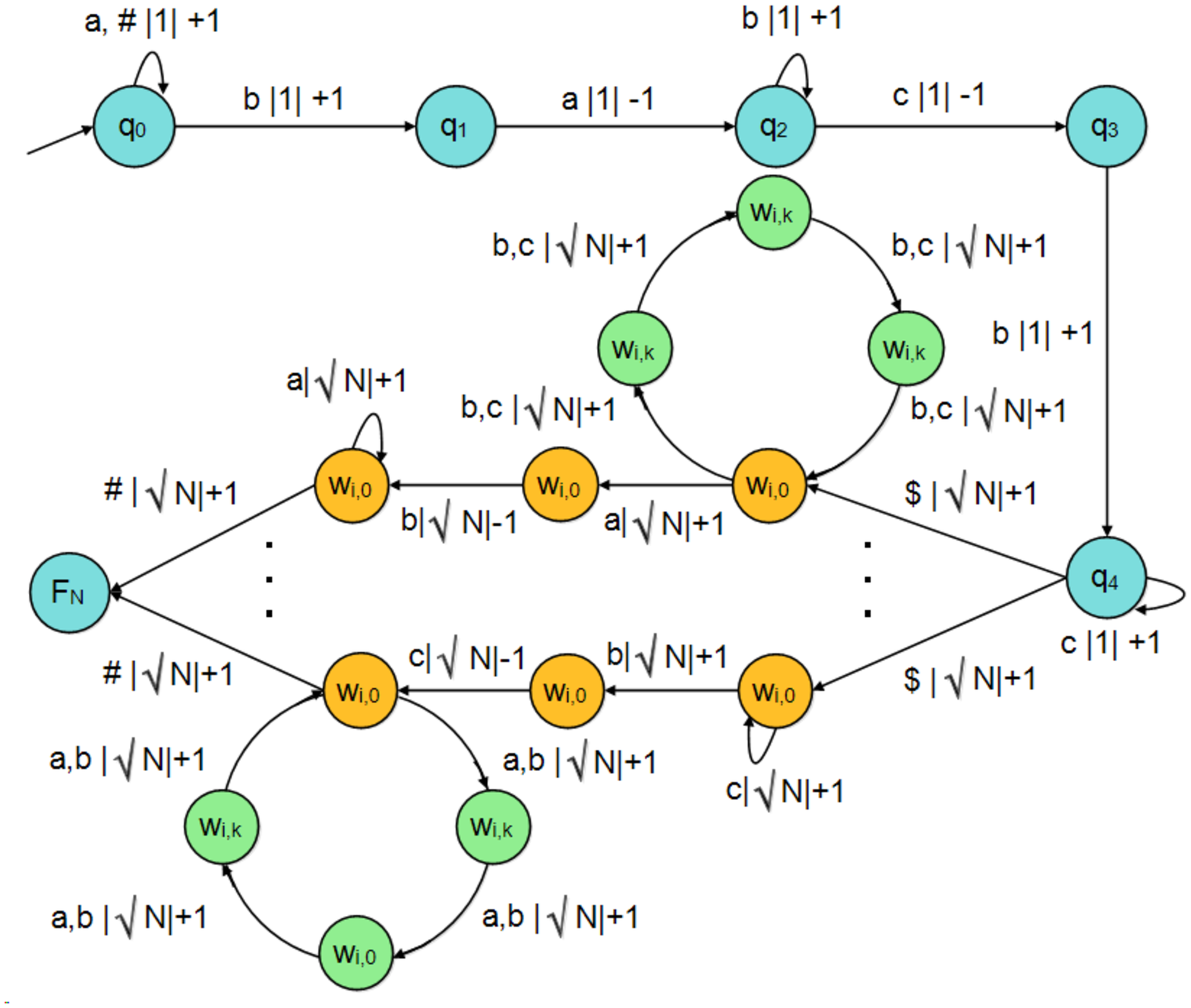}
	\caption{State transition diagram of $L_3$}
\end{figure}  
\begin{proof}
The design of proof for BZ reaction network is as follows. It consists of two phases. First, the 2QFA traverse the input to check the form $a^+b^+c^+$. On reading the right-end marker \$, the computation is split into \textit{N} paths such that $w_{1,0}, w_{2,0},..., w_{N,0}$. Secondly, the first path is used to check the number of $b$'s and $c$'s are equal or not. The second path is used to check the initial part of an input string to identify it is in $\{a^n b^n \mid n > 0\}$. On reading the right-end marker \$, the both paths are split into \textit{N} different paths with an equal amplitude $\dfrac{1}{\sqrt{N}}$.  Finally, upon reading the right-end marker \#, 
 if the number of \textit{a}'s and \textit{b}'s and \textit{b}'s and \textit{c}'s are equal in respective computational paths. Then, all paths come into \textit{N}-way Quantum Fourier transform (QFT) and either one acceptance state or rejectance states are observed. Suppose, if the input string is not in the corrected form, then all computation paths read the \# at different times. Thus, their amplitudes do not cancel each other and the input string is said to be rejected with probability $1-\dfrac{1}{N}$. Otherwise, the input string is said to be recognized by 2QFA with probability 1. 
 \end{proof}

\section{Summary}
In summary, 2QFA model can be efficiently designed for balanced chemical reaction and Belousov-Zhabotinsky (BZ) reaction network with one-sided error bound, which are halted in linear time. Table 3 shows the language recognition ability of different computational models. The classical 2DFA and 2PFA are known to be equal in computational power to one-way deterministic finite automata (1DFA) \cite{45, 46}. It has been proved that 2PFA's can be designed for non-regular languages in expected polynomial time. Additionally, it has been demonstrated that the chemical PDA can be designed for aforementioned chemical reactions with multiple stacks. The recognition of languages by native chemical automata cab be found in \cite{12, 13, 14}. But, we have shown that 2QFA can recognize such chemical reactions, without any external aid. It has been proved that 2QFA is more powerful than classical variants, because it follows the quantum superposition principle to be in more than one state at a time on the input tape. For the execution, it needs atleast $O(log~n)$ quantum states to store the position of tape head, where \textit{n} denotes the length of an input string. 

\begin{table}[!ht]
\begin{threeparttable}
\centering 
%\makegapedcells
\caption{Comparison of computational power of models} 
\begin{tabular}{|l|c|c|c|c|}
	\hline
	Languages & Class & 2DFA/2PFA & Chemical FA/PDA & 2QFA \\
	\hline
	\makecell{$L_{1}=\{(a+b)^*a(a+b)^*b$\\  $(a+b)^*aa^*bb^*\}$} & RL & \cmark & \cmark & \cmark \\
	\hline
	 \makecell{A language $L_{2}$ consisting Dyck \\ language of balanced parentheses} & CFL & \xmark & \makecell{\cmark \\ (with 1-stack PDA)} & \cmark \\
	\hline
	\makecell{$L_{3}=\{a^nb^nc^n | n > 0\}$} & CSL & \xmark & \makecell{\cmark \\ (with 2-stack PDA)} & \cmark \\
	\hline
\end{tabular}
 \begin{tablenotes}
      \footnotesize	
      \item RL, CFL and CSL stand for regular languages, context-free languages and context-sensitive languages, respectively.
    \end{tablenotes}
  \end{threeparttable}
\end{table}

\section{Conclusion}
The enhancement in many existing computational approaches provides momentum to molecular and quantum simulations at the atomic level. It helps to test new abstract approaches for considering molecules and matter. Previous attempts to model the aforementioned chemical reactions used finite automata and pushdown automata with multiple stacks. In this paper, we focused on well-known languages of Chomsky hierarchy and modeled them using two-way quantum finite automata. The crucial advantage of the quantum approach is that these chemical reactions transcribed in formal languages can be parsed in linear time, without using any external aid. We have shown that two-way quantum automata are more superior than its classical variants by using quantum transitions. To the best of our knowledge, no such modeling of chemical reactions is performed using quantum automata theory so far. For the future purpose, we will try to represent complex chemical reactions in formal languages and model them using other quantum computational models.

\section{Additional Information}

\textbf{Conflict of interest} The authors declare that they have no conflict of
interest.

\section*{Acknowledgement}
S.Z. acknowledges support in part from the National Natural Science
Foundation of China (Nos. 61602532),  the Natural Science Foundation of Guangdong Province of China (No. 2017A030313378), and the Science and Technology Program of Guangzhou City of China (No. 201707010194).

\end{document}